\documentclass[1p,final]{elsarticle}
\usepackage{amsfonts,color,morefloats,pslatex}
\usepackage{amssymb,amsthm, amsmath,latexsym}

\allowdisplaybreaks[4]

\newtheorem{rem}{Remark}
\newtheorem{theorem}{Theorem}
\newtheorem{lemma}[theorem]{Lemma}

\newtheorem{corollary}[theorem]{Corollary}

\newtheorem{example}{Example}

\newcommand{\gf}{{\mathrm{GF}}}

\newcommand{\support}{{\mathrm{Suppt}}}

\newcommand{\re}{{\mathtt{Re}}}
\newcommand{\wt}{{\mathtt{wt}}}

\newcommand{\w}{{\mathtt{w}}}

\newcommand{\C}{{\mathcal{C}}}

\newcommand{\ba}{{\mathbf{a}}}
\newcommand{\bb}{{\mathbf{b}}}
\newcommand{\bc}{{\mathbf{c}}}

\newcommand{\bbf}{{\mathbf{f}}}
\newcommand{\bbv}{{\mathbf{v}}}
\newcommand{\bbu}{{\mathbf{u}}}

\newcommand{\bzero}{{\mathbf{0}}}

\begin{document}

\begin{frontmatter}



\title{Minimal Linear Codes over Finite Fields\tnotetext[fn1]{C. Ding's research was supported by
The Hong Kong Research Grants Council, Proj. No. 16300415.
Z. Zhou was supported by
the Natural Science Foundation of China under Grants  11571285  and 61672028, and also the Sichuan Provincial Youth Science and Technology Fund under Grant 2015JQ0004 and 2016JQ0004. }}


\author[cding]{Ziling Heng}
\ead{zilingheng@163.com} 
\author[cding]{Cunsheng Ding}
\ead{cding@ust.hk}
\author[zz]{Zhengchun Zhou}
\ead{zzc@home.swjtu.edu.cn}

\cortext[zz]{Z. Zhou is the Corresponding author}

\address[cding]{Department of Computer Science
                                                  and Engineering, The Hong Kong University of Science and Technology,
                                                  Clear Water Bay, Kowloon, Hong Kong, China}                                                  
\address[zz]{School of Mathematics, Southwest Jiaotong University,
Chengdu, 610031, China}

\begin{abstract}
As a special class of linear codes, minimal linear codes have important applications in secret sharing and secure two-party computation. Constructing minimal linear codes with new and desirable parameters has been an interesting research topic in coding
theory and cryptography. Ashikhmin and Barg showed that  $w_{\min}/w_{\max}> (q-1)/q$ is a sufficient condition  for a linear code over the finite field $\gf(q)$ to be minimal, where $q$ is a prime power, $w_{\min}$ and $w_{\max}$ denote the minimum and maximum nonzero weights in the code, respectively. The first objective of this paper is to present a  sufficient and necessary condition for linear codes over finite fields to be minimal. 
The second objective of this paper is to construct an infinite family of ternary minimal linear codes satisfying $w_{\min}/w_{\max}\leq  2/3$. To the best of our knowledge, this is the first infinite family of nonbinary minimal linear codes violating  Ashikhmin and Barg's condition.
\end{abstract}


\end{frontmatter}

\section{Introduction}

Let $q$ be a prime power and $\gf(q)$ denote the finite field with $q$ elements.
An  $[n,\, k,\, d]$ code $\C$ over $\gf(q)$ is a $k$-dimensional subspace of $\gf(q)^n$ with minimum
(Hamming) distance $d$.
Let $A_i$ denote the number of codewords with Hamming weight $i$ in a code
$\C$ of length $n$. The {\em weight enumerator} of $\C$ is defined by
$
1+A_1z+A_2z^2+ \cdots + A_nz^n.
$
The sequence $(1, A_1, A_2, \cdots, A_n)$ is called the \emph{weight distribution} of the code $\C$.

The support of a vector $\bbv=(v_1,v_2,\ldots,v_n) \in \gf(q)^n$, denoted by $\support(\bbv)$, is defined by
$$
\support(\bbv)=\{1 \leq i \leq n: v_i \neq 0 \}.
$$
The vector $\bbv$ is called the characteristic vector or the incidence vector of the set $\support(\bbv)$.
A vector $\bbu \in \gf(q)^n$ covers another vector $\bbv \in \gf(q)^n$ if $\support(\bbu)$ contains $\support(\bbv)$.
We write $\bbv \preceq \bbu$ if $\bbv$ is covered by $\bbu$, and $\bbv \prec \bbu$ if $\support(\bbv)$ is a proper
subset of $\support(\bbu)$. A codeword $\bbu$ in a linear code $\C$ is said to be \emph{minimal} if $\bbu$ covers only
the codeword $a\bbu$ for all $a \in \gf(q)$, but no other codewords in $\C$. A linear code
$\C$ is said to be \emph{minimal} if every codeword in $\C$ is minimal.

Minimal linear codes have interesting applications in secret sharing \cite{CDY05, Massey93,YD06} and secure two-party computation \cite{ABCH95,CMP03}, and could be decoded with a minimum distance decoding method \cite{AB98}.
Searching for minimal linear codes has been an interesting research topic in coding
theory and cryptography. The following sufficient condition
for a linear code to be minimal is due to Ashikhmin and Barg \cite{AB98}.

\begin{lemma}[Ashikhmin-Barg]\label{lem-AB}
A linear code $\C$ over $\gf(q)$ is minimal if 
$$
\frac{w_{\min}}{w_{\max}}> \frac{q-1}{q},$$ 
where $w_{\min}$ and $w_{\max}$ denote the minimum and maximum nonzero Hamming weights in
the code $\C$, respectively.
\end{lemma}

With the help of Lemma \ref{lem-AB}, a number of families of minimal linear codes
with  $w_{\min}/w_{\max}> (q-1)/q$ have been reported in the literature (see, \cite{CDY05}, \cite{Ding15}, \cite{DLLZsurvey}, \cite{YD06}, for example). Sporadic examples
in \cite{CMP03} show that Ashikhmin-Barg's condition is not necessary for linear codes to be minimal.
However, no infinite family of minimal linear codes with $w_{\min}/w_{\max} \leq (q-1)/q$
was found until the  breakthrough in \cite{CH17}, where an infinite family of such
binary codes was discovered. Inspired by the work in \cite{CH17}, the authors of the present paper
gave a further study of binary minimal linear codes \cite{DHZ}. Specifically,
a necessary and sufficient condition for binary linear codes to be minimal was derived 
in \cite{DHZ}. With this
new condition, three infinite families
of minimal binary linear codes with $w_{\min}/w_{\max} \leq 1/2$ were obtained from a general construction in \cite{DHZ}.

The first objective of this paper is to present a sufficient and necessary condition for  linear codes over finite fields to be minimal, which generalizes the result about the binary case given in \cite{DHZ}. The second objective of this paper is to construct an infinite family of ternary minimal linear codes with  $w_{\min}/w_{\max}< 2/3$. To the best of our knowledge, this is the first
infinite family of nonbinary minimal linear codes violating the Ashikhmin-Barg condition.

The rest of this paper is organized as follows. In Section \ref{sec-krawtchouk}, we recall basic properties of Krawtchouk polynomials which will be needed in the sequel.
In Section \ref{sec-condition}, we present a new sufficient and necessary condition for linear codes over finite fields to be minimal. In Section \ref{sec-construction},
we use the Walsh spectrum of generalized Boolean functions to characterize when ternary linear codes from a general construction are minimal. We then propose a family of ternary minimal codes violating the Ashikhmin-Barg condition with this characterization. Finally, we conclude this paper and make concluding comments in Section \ref{sec-concluison}.

\section{Krawtchouk polynomials and their properties}\label{sec-krawtchouk}

Krawtchouk polynomials were introduced by Lloyd in 1957 \cite{L1957} and have wide  applications in coding theory \cite{B,FKLW,HP}, cryptography \cite{CZLH09}, and combinatorics \cite{Leven}. In this section, we only give a short introduction to Krawtchouk polynomials with their essential properties. For more information, the reader is referred to \cite{B, FKLW, Leven, L1957}.

Let $m$ be a positive integer, $q$ a positive integer and $x$ a variable taking nonnegative values. The  \emph{Krawtchouk polynomial} (of degree $t$ and with parameters $q$ and $m$) is defined by
\begin{eqnarray*}
K_t(x,m)=\sum_{j=0}^t (-1)^j(q-1)^{t-j} \binom{x}{j} \binom{m-x}{t-j}.
\end{eqnarray*}
Accordingly, the \emph{Lloyd polynomial} $\Psi_{k}(x,m)$ (of degree $k$ and with parameters $q$ and $m$) is given by
\begin{eqnarray}\label{eqn-lylod-general}
\Psi_{k}(x,m)=\sum_{t=0}^{k}K_t(x,m).
\end{eqnarray}

The following results will be useful in the sequel.

\begin{lemma}\cite[Lemma 3.2.1]{B}\label{lem-Lloydpolynomial} For $x,m\geq 1$,
$$
\Psi_{k}(x,m)=K_k(x-1,m-1).
$$
\end{lemma}

\begin{lemma}\label{thm-ms21}
Let symbols and notation be as before. Then the following holds: 
\begin{itemize}
\item[(1)] $K_t(0,m)=(q-1)^{t}\binom{m}{t}.$
\item[(2)] $K_t(1,m)=(q-1)^{t}\binom{m-1}{t}-(q-1)^{t-1}\binom{m-1}{t-1}.$
\item[(3)] $K_t(m,m)=(-1)^{t}\binom{m}{t}.$
\item[(4)] \cite[Lemma 3.3.1]{B} For any integers $0\leq x,t\leq m$, $$(q-1)(m-x)K_t(x+1,m)-(x+(q-1)(m-x)-qt)K_t(x,m)+xK_t(x-1,m)=0.$$
\item[(5)] \cite[Equation (21)]{Leven} For any integers $0\leq x,t\leq m$,
$$(q-1)^{x}\binom{m}{x}K_t(x,m)=(q-1)^{t}\binom{m}{t}K_x(t,m).$$
\end{itemize}
\end{lemma}

\begin{lemma}\cite[Equation (6)]{FKLW}\label{lem-bound1}
For $x,t\in \{0,1,2,\cdots,m\}$,
$$
K_t(x,m)\leq K_t(0,m).
$$
\end{lemma}

\begin{lemma}\label{lem-bound2}
For $x,t\in \{0,1,2,\cdots,m\}$,
$$
|K_t(x,m)|\leq (q-1)^{t}\binom{m}{t}.
$$
\end{lemma}
\begin{proof}
For any given $x,t\in \{0,1,2,\cdots,m\}$, we have
\begin{eqnarray*}
|K_t(x,m)|&=&\left|\sum_{j=0}^t (-1)^j(q-1)^{t-j} \binom{x}{j} \binom{m-x}{t-j}\right|\\
&\leq & \sum_{j=0}^t \left|(-1)^j(q-1)^{t-j} \binom{x}{j} \binom{m-x}{t-j}\right|\\
&\leq & (q-1)^{t}\sum_{j=0}^t  \binom{x}{j} \binom{m-x}{t-j}\\
&= &(q-1)^{t}\binom{m}{t},
\end{eqnarray*}
where the last equality followed from the Vandermonde convolution formula.
This completes the proof of this lemma.
\end{proof}


The following result follows directly from Lemmas \ref{lem-Lloydpolynomial} and \ref{lem-bound2}.

\begin{corollary}\label{cor1}
For any integers $1\leq x\leq m$ and $1\leq k \leq m-1$, we have
$$
|\Psi_{k}(x,m)|\leq (q-1)^{k}\binom{m-1}{k}.
$$
\end{corollary}


We remark that the upper bound for $|\Psi_{k}(x,m)|$ in Corollary \ref{cor1} is tight since
$$
\Psi_{k}(1,m)=K_{k}(0,m-1)=(q-1)^{k}\binom{m-1}{k}.
$$

The next result will be employed to calculate the Hamming weights of the proposed linear codes in Section \ref{sec-construction}.
\begin{lemma}\cite[Lemma 4.2.1]{B}\label{lem-krawweight}
Let $\bbu \in \mathbb{Z}_q^m$ with Hamming weight $\wt(\bbu)=i$. Then
$$
\sum_{\substack{\bbv \in \mathbb{Z}_q^m\\ \wt(\bbv)=t}} \zeta_{q}^{\bbu \cdot \bbv} = K_t(i,m),
$$
where $\zeta_q$ denotes the $q$-th primitive root of complex unity, and the inner product $\bbu\cdot \bbv$ in $\mathbb{Z}_q^m$ is defined by  $\bbu\cdot \bbv=u_1v_1+\cdots+u_mv_m$.
\end{lemma}

\section{A sufficient and necessary condition for  $q$-ary linear codes to be minimal}\label{sec-condition}

In this section, we shall present a sufficient and necessary condition for   linear codes over $\gf(q)$ to be minimal.
From now on, we always assume that $q$ is a prime power. Let $\gf(q)^*=\gf(q)\setminus \{0\}$ and
$(\gf(q)^{n})^*=\gf(q)^{n}\setminus \{\bzero\}$.
For any $\ba=(a_1,a_2,\ldots,a_n),\bb=(b_1,b_2,\ldots,b_n)\in \gf(q)^{n}$, define $\ba\cap \textbf{b}$ to be the vector $(f_1,f_2,\ldots,f_n)$ as
\begin{eqnarray*}
f_i=\left\{
\begin{array}{ll}
a_i,   &      \mbox{if }a_i=b_i\in \gf(q)^{*},\\
0, & \mbox{otherwise},\\
\end{array} \right.
\end{eqnarray*}
For example, for $\ba=(1,0,2,1,0),\bb=(2,1,2,0,0)\in\gf(3)^{5}$, we have
$$
\ba\cap \bb=(0,0,2,0,0).
$$

The following three lemmas will be needed in the sequel.

\begin{lemma}\label{lem-iff}
For any $\ba,\bb\in \gf(q)^n$, $\bb\preceq \ba$ if and only if
\begin{eqnarray}\label{eqn-iff1}
\sum_{c\in \gf(q)^*}(c\ba\cap \bb)=\bb
\end{eqnarray}
if and only if
\begin{eqnarray}\label{eqn-iff2}\sum_{c\in \gf(q)^*}\wt(c\ba\cap \bb)=\wt(\bb).\end{eqnarray}
\end{lemma}
\begin{proof}
Let  $\ba=(a_1,a_2,\ldots,a_n)$ and $\bb=(b_1,b_2,\ldots,b_n)$ be any two vectors
in $\gf(q)^{n}$. We first prove that $\bb\preceq \ba$ if and only if (\ref{eqn-iff1}) holds.
Assume that $\bb\preceq \ba$, then for any  $1\leq i \leq n$,  $b_i\neq 0$ implies $a_i\neq 0$ and
there is one and only one $c\in \gf(q)^*$ such that $b_i=ca_i$. This together with the definition of
$\ba \cap \bb$  leads to (\ref{eqn-iff1}). On the other hand, suppose that  (\ref{eqn-iff1}) holds, then for any $1\leq i \leq n$, these exists some $c\in \gf(q)^*$ such that $b_i=ca_i$. Therefore $b_i\neq 0$ implies $a_i\neq 0$ and further implies $\bb\preceq \ba$. The conclusion that (\ref{eqn-iff1}) holds if and only if (\ref{eqn-iff2})
holds follows directly from the fact
that
$$
\support(c_1\ba\cap \bb)\cap\support(c_2\ba\cap \bb)=\emptyset
$$
for any distinct pair $(c_1,c_2)\in \gf(q)^*\times \gf(q)^*$.
\end{proof}

\begin{lemma}\label{lemma-key}
For any $\ba,\bb\in \gf(q)^{n}$,
\begin{eqnarray}\label{eqn-key}
(q-1)(\wt(\ba)+\wt(\bb))=\sum_{c\in \gf(q)^*}\wt(\ba+c\bb)+q\sum_{c\in \gf(q)^*}\wt\left(c\ba \cap \bb\right).
\end{eqnarray}
\end{lemma}
\begin{proof}
For any $c\in \gf(q)^*$, with a direct verification, we have
\begin{eqnarray*}\label{eqn-1}
\wt(\ba)+\wt(\bb)=\wt(\ba+c\bb)+\sum_{y\in \gf(q)^*}\wt\left(y\ba \cap \bb\right)+\wt\left(-\frac{1}{c} \ba\cap \bb\right).
\end{eqnarray*}
It then follows that
\begin{eqnarray*}\label{eqn-2}
(q-1)(\wt(\ba)+\wt(\bb))
&=&\sum_{c\in \gf(q)^*}\wt(\ba+c\bb)+(q-1)\sum_{y\in \gf(q)^*}\wt\left(y\ba \cap \bb\right)\\
& &+\sum_{c\in \gf(q)^*}\wt\left(c \ba\cap \bb\right) \nonumber\\
&=&\sum_{c\in \gf(q)^*}\wt(\ba+c\bb)+q\sum_{c\in \gf(q)^*}\wt\left(c\ba \cap \bb\right).
\end{eqnarray*}
This comoletes the proof of this lemma.
\end{proof}

\begin{lemma}\label{lem-sn-condition}
For any $\ba,\bb\in \gf(q)^{n}$, $\bb\preceq \ba$ if and only if
$$
\sum_{c\in \Bbb \gf(q)^*}\wt(\ba+c\bb)=(q-1)\wt(\ba)-\wt(\bb).
$$
\end{lemma}
\begin{proof}
The conclusion follows from Lemmas \ref{lem-iff} and \ref{lemma-key} by plugging Equation (\ref{eqn-iff2}) into Equation (\ref{eqn-key}).
\end{proof}

We are now in a position to present the sufficient and necessary condition for  linear codes over $\gf(q)$ to be minimal.

\begin{theorem}\label{thm-main1}
Let $\mathcal{C}\subseteq \gf(q)^{n}$ be a linear code. Then it is minimal if and only if
$$
\sum_{c\in \Bbb \gf(q)^*}\wt(\ba+c\bb)\neq(q-1)\wt(\ba)-\wt(\bb)
$$
for any $\gf(q)$-linearly independent codewords $\ba,\bb\in \mathcal{C}$.
\end{theorem}
\begin{proof}
The conclusion follows from the definition of minimal codes and  Lemma \ref{lem-sn-condition}.
\end{proof}

\begin{rem}
When $q=2$, according to Theorem \ref{thm-main1}, $\C$ is minimal if and only if
$$
\wt(\ba+\bb)\neq\wt(\ba)-\wt(\bb)
$$
for any two distinct codewords $\ba,\bb\in \mathcal{C}$.
This sufficient and necessary condition was derived in \cite{DHZ}. Therefore, the result in
Theorem \ref{thm-main1} generalizes the one in \cite{DHZ} since it works for
any prime power $q$.
\end{rem}

By Theorem \ref{thm-main1}, the minimality of $\mathcal{C}$ is completely determined by the weights of its codewords. In particular, if $\mathcal{C}$ is a $q$-ary linear code with only two weights, by Theorem \ref{thm-main1}, we can judge the minimality of it as follows.

\begin{corollary}\label{cor-twoweight}
Let $\mathcal{C}\subseteq \gf(q)^{n}$ be a two-weight $q$-ary linear code with nonzero weights $w_1$ and $w_2$, where $0<w_1<w_2< n$. Then  $\mathcal{C}$ is minimal, provided that
\begin{eqnarray*}
jw_1\neq (j-1)w_2
\end{eqnarray*}
for any integer $j$ with $2\leq j\leq q$.
\end{corollary}

\begin{proof} 
Suppose that $\mathcal{C}$ is not minimal. 
By Theorem \ref{thm-main1}, there exists a pair of $\gf(q)$-linearly independent codewords $\ba,\bb\in \mathcal{C}$ such
that
\begin{eqnarray}\label{eqn-6}
\wt(\textbf{b}) + \sum_{c\in \Bbb \gf(q)^*}\wt(\textbf{a}+c\textbf{b}) 
=(q-1)\wt(\textbf{a}).
\end{eqnarray} 
Note that $\wt(\textbf{a}+c\textbf{b})>0$ for any $c\in \Bbb \gf(q)^*$, as $\wt(\ba)>0$ and $\wt(\bb)>0$.
Thus we have $\wt(\ba)=w_2$. Consider the multiset
$$
\{\wt(\textbf{a}+c\textbf{b}):c\in\Bbb \gf(q)^* \}\cup\{\wt(\textbf{b})\}.
$$
Assume that the multiplicity of $w_1$ in this multiset is $j$ and the multiplicity of $w_2$ is $q-j$. Then we have
$2\leq j\leq q$. It follows from (\ref{eqn-6}) that
$$
jw_1=(j-1)w_2,
$$
This completes the proof.
\end{proof}

\section{A family of minimal ternary linear codes violating the Ashikhmin-Barg condition}\label{sec-construction}

In this section, we present a family of minimal ternary linear codes violating the Ashikhmin-Barg condition with a general construction. The idea is similar to the 
one in our construction for the binary case \cite{DHZ}.

\subsection{A general construction of ternary linear codes}

Throughout this section, we always assume that $f(x)$ is a function from $\gf(3)^m$ to $\gf(3)$ such that $f(\bzero)=0$ but $f(b)\neq 0$ for at least
one $b \in \gf(3)^m$. Recall that the Walsh transform of $f$  is given
as
\begin{eqnarray*}
\hat{f}(w)=\sum_{x \in \gf(3)^m} \zeta_{3}^{f(x)-w \cdot x}, ~w\in \gf(3)^m,
\end{eqnarray*}
where $\zeta_{3}$ is a primitive $3$-th  complex root of unity, and  $w\cdot x$ denotes the standard inner product of $w$ and $x$. Using such function $f$, we define a linear code by
\begin{eqnarray}\label{eqn-mycode}
\C_f=\{(uf(x)+v \cdot x)_{x \in (\gf(3)^m)^*}: u \in \gf(3), \ v \in \gf(3)^m\}.
\end{eqnarray}

The construction above is general in the sense that it works for any function $f$ from $\gf(3)^m$ to $\gf(3)$ with $f(\bzero)=0$.
The following result shows that the weight distribution of $\C_f$ could be determined by the Walsh spectrum of $f$.

\begin{theorem}\label{thm-codeparam}
Assume that $f(x)\neq w \cdot x$  for any $w\in \gf(3)^{m}$. The linear code $\C_f$ in (\ref{eqn-mycode}) has length $3^m-1$ and dimension $m+1$. In addition,
the weight distribution of $\C_f$ is given by the following multiset union:
\begin{eqnarray}\label{multiset}
&& \left\{\left\{2\left(3^{m-1}-\frac{\re\left(\hat{f}(v)\right)}{3}\right): u \in \gf(3)^*, v \in \gf(3)^m \right\}\right\} 
\cup \nonumber \\ 
&& \left\{\left\{3^{m}-3^{m-1}: u=0,v \in (\gf(3)^m)^* \right\}\right\} \cup \{\{0\}\}.
\end{eqnarray}
Herein and hereafter,  $\re(x)$ denotes the real part of the complex number $x$.
\end{theorem}

\begin{proof} 
In terms of exponential sums, the Hamming weight $\wt(\bc)$ of any codeword $\bc=(uf(x)+v \cdot x)_{x \in \gf(3)^m \setminus \{\bzero\}}$ of $\C_f$ in (\ref{eqn-mycode}) can be calculated as
\begin{eqnarray*}
\wt(\bc)&=&\sharp \{x\in \gf(3)^m \setminus \{\bzero\}:uf(x)+v \cdot x\neq 0\}\\
&=&(3^{m}-1)-\frac{1}{3}\sum_{y\in \gf(3)}\sum_{x\in \gf(3)^m \setminus \{\bzero\}}\zeta_{3}^{y(uf(x)+v \cdot x)}\\
&=&(3^{m}-1)+1-3^{m-1}-\frac{1}{3}\sum_{y\in \gf(3)^{*}}\sum_{x\in \gf(3)^m}\zeta_{3}^{y(uf(x)+v \cdot x)}\\
&=&3^{m}-3^{m-1}-\frac{1}{3}\left(\sum_{x\in \gf(3)^m}\zeta_{3}^{uf(x)+v \cdot x}+\sum_{x\in \gf(3)^m}\zeta_{3}^{-uf(x)-v \cdot x}\right)\\
&=&3^{m}-3^{m-1}-\frac{2}{3}\re \left(\sum_{x\in \gf(3)^m}\zeta_{3}^{uf(x)+v \cdot x}\right).
\end{eqnarray*}
We discuss the value of $\wt(\bc)$ by considering the following cases.
\begin{enumerate}
\item[$\bullet$] If $u=0$ and $v=\textbf{0}$, then $\wt(\bc)=0$.
\item[$\bullet$] If $u=0$ and $v\neq\textbf{0}$, then $\wt(\bc)=3^{m}-3^{m-1}$.
\item[$\bullet$] If $u=1$ and $v\neq\textbf{0}$, then $\wt(\bc)=3^{m}-3^{m-1}-\frac{2}{3}\re(\hat{f}(-v))$.
\item[$\bullet$] If $u=-1$ and $v\neq\textbf{0}$, then $\wt(\bc)=3^{m}-3^{m-1}-\frac{2}{3}\re(\hat{f}(v))$.
\item[$\bullet$] If $u\in \gf(3)^{*}$ and $v=\textbf{0}$, then $\wt(\bc)=3^{m}-3^{m-1}-\frac{2}{3}\re(\hat{f}(\textbf{0}))$.
\end{enumerate}
The weight distribution in (\ref{multiset}) then follows from the discussions above. Note that the dimension of the code $\C_f$ is $m+1$ if and only if
$$\re\left(\hat{f}(w)\right)\neq 3^{m}\mbox{ for any }w \in \gf(3)^m,$$
where
\begin{eqnarray*}
\re\left(\hat{f}(w)\right)&=&\re\left(\sum_{x \in \gf(3)^m} \zeta_{3}^{f(x)-w \cdot x}\right) 
= \sum_{x \in \gf(3)^m} \re\left(\zeta_{3}^{f(x)-w \cdot x}\right).
\end{eqnarray*}
Hence $\re\left(\hat{f}(w)\right)= 3^{m}$ if and only if $f(x)=w\cdot x$ for all $x\in \gf(3^m)$.
This together with the hypothesis that $f(x)\neq w \cdot x$  for any $w\in \gf(3)^{m}$ means that the
dimension of  $\C_f$ is $m+1$.
\end{proof}

\subsection{When are these codes minimal?}

Now, a natural question is when the linear code $\C_f$ defined in (\ref{eqn-mycode}) is minimal.
The following gives a sufficient and necessary condition for $\C_f$ to be minimal in terms of the Walsh spectrum of $f$.

\begin{theorem}\label{thm-1stconminimal}
Let $\C_f$ be the ternary code of Theorem \ref{thm-codeparam}. Assume that $f(x)\neq v \cdot x$ for any $v\in \gf(3)^{m}$. Then $\C_f$ is a minimal $[3^{m}-1,m+1]$ code if and only if
\begin{eqnarray*}
\re(\hat{f}(w_1))+\re(\hat{f}(w_2))-2\re(\hat{f}(w_3))\neq3^{m}
\end{eqnarray*}
and
\begin{eqnarray*}
\re(\hat{f}(w_1))+\re(\hat{f}(w_2))+\re(\hat{f}(w_3))\neq3^{m}
\end{eqnarray*}
for any pairwise distinct vectors $w_1,w_2$ and $w_3$ in $\gf(3)^m$ satisfying $w_1+w_2+w_3=\bzero$.
\end{theorem}

\begin{proof}
We define the following linear code
$$\mathcal{S}_{m}=\{(v\cdot x)_{x\in\gf(3)^{m}\backslash\{\bzero\}}:v\in \gf(3)^{m}\}.$$
This code is a ternary code with parameters $[3^m-1,m,3^m-3^{m-1}]$ and the only nonzero Hamming weight $3^m-3^{m-1}$.

Assume that $f(x)\neq v \cdot x$ for any $v\in \gf(3)^{m}$. Let $\bbf=(f(x))_{x\in\gf(3)^{m}\backslash\{\bzero\}}$. By definition, every codeword $a\in \C_f$ can be expressed as
$$a=u_a\bbf+s_a,$$
where $u_a\in \{0,1,-1\}$ and $s_a=(v_a\cdot x)_{x\in \gf(3)^{m} \setminus \{\bzero\}}\in \mathcal{S}_{m}$ for some $v_a\in \gf(3)^{m}$. We next consider the
coverage of codewords in $\C_f$ by distinguishing the following cases.

Case I: Let $a=s_a$ and $b=s_b$ be two linearly independent codewords in $\mathcal{S}_{m}$. Then one cannot cover
the other as the one-weight code $\mathcal{S}_{m}$ is obviously minimal.

Case II: Let $a=\bbf+s_a$ and $b=\bbf+s_b$, where $a,b$ are linearly independent. This implies that
$$a\pm b\neq \bzero,\mbox{ i.e., }s_a\neq s_b\mbox{ and }s_a+s_b\neq \bbf.$$
By assumption, $s_a+s_b \neq \bbf$ always holds. Hence $a,b$ are linearly independent if and only if $s_a\neq s_b$. Since
$$s_a=(v_a\cdot x)_{x\in\gf(3)^{m}\setminus \{\bzero\}},\ 
s_b=(v_b\cdot x)_{x\in\gf(3)^{m} \setminus \{\bzero \}},$$ we further deduce that $a,b$ are linearly independent if and only if $v_a\neq v_b$.
Suppose that $b\preceq a$. It then follows from Lemma \ref{lem-sn-condition} and the proof of Theorem \ref{thm-codeparam} that
\begin{eqnarray*}
b\preceq a &\iff & \wt(\ba+\bb)+\wt(\ba-\bb)=2\wt(\ba)-\wt(\bb)\\
&\iff & \left(3^{m}-3^{m-1}-\frac{2}{3}\re(\hat{f}(v_a+v_b))\right)+(3^{m}-3^{m-1})\\
& &=2\left(3^{m}-3^{m-1}-\frac{2}{3}\re(\hat{f}(-v_a))\right)-\left(3^{m}-3^{m-1}-\frac{2}{3}\re(\hat{f}(-v_b))\right)\\
& \iff & \re(\hat{f}(v_a+v_b))-2\re(\hat{f}(-v_a))+\re(\hat{f}(-v_b))=3^{m}.
\end{eqnarray*}
Similarly, we have
$$a\preceq b \iff \re(\hat{f}(v_a+v_b))-2\re(\hat{f}(-v_b))+\re(\hat{f}(-v_a))=3^{m}.$$
Note that $v_a+v_b,-v_a,-v_b$ are pairwise distinct as $v_a\neq v_b$. In addition,
$$v_a+v_b-v_a-v_b=\bzero.$$

Case III: Let $a=-\bbf+s_a$ and $b=-\bbf+s_b$, where $a,b$ are linearly independent. Similarly as in Case II, we deduce that $v_a\neq v_b$ and
$$b\preceq a \iff \re(\hat{f}(-v_a-v_b))-2\re(\hat{f}(v_a))+\re(\hat{f}(v_b))=3^{m}$$ and
$$a\preceq b \iff \re(\hat{f}(-v_a-v_b))-2\re(\hat{f}(v_b))+\re(\hat{f}(v_a))=3^{m}.$$
Note that $-v_a-v_b,v_a,v_b$ are pairwise distinct as $v_a\neq v_b$. In addition,
$$-v_a-v_b+v_a+v_b=\bzero.$$

Case IV: Let $a=s_a$ and $b=\bbf+s_b$ with $a$ being nonzero. Then $a,b$ are linearly independent because of the assumption. Similarly as in Case II, we deduce that
$$b\preceq a \iff \re(\hat{f}(-v_a-v_b))+\re(\hat{f}(v_a-v_b))+\re(\hat{f}(-v_b))=3^{m}$$ and
$$a\preceq b \iff \re(\hat{f}(-v_a-v_b))+\re(\hat{f}(v_a-v_b))-2\re(\hat{f}(-v_b))=3^{m}.$$
Note that $-v_a-v_b,v_a-v_b,-v_b$ are pairwise distinct as $v_a\neq \bzero$. In addition,
$$-v_a-v_b+v_a-v_b-v_b=\bzero.$$

Case V: Let $a=s_a$ and $b=-\bbf+s_b$ with $a$ being nonzero. Then $a,b$ are linearly independent because of the assumption. Similarly as in Case II, we deduce that
$$b\preceq a \iff \re(\hat{f}(v_a+v_b))+\re(\hat{f}(v_b-v_a))+\re(\hat{f}(v_b))=3^{m}$$ and
$$b\preceq a \iff \re(\hat{f}(v_a+v_b))+\re(\hat{f}(v_b-v_a))-2\re(\hat{f}(v_b))=3^{m}.$$
Note that $v_a+v_b,v_b-v_a,v_b$ are pairwise distinct as $v_a\neq \bzero$. In addition,
$$v_a+v_b+v_b-v_a+v_b=\bzero.$$

Case VI: Let $a=\bbf+s_a$ and $b=-\bbf+s_b$ with $s_a\neq -s_b$. Then $a,b$ are linearly independent because of the assumption. Similarly as in Case II, we deduce that
$$b\preceq a \iff \re(\hat{f}(v_a-v_b))+\re(\hat{f}(v_b))-2\re(\hat{f}(-v_a))=3^{m}$$ and
$$a\preceq b \iff \re(\hat{f}(v_a-v_b))+\re(\hat{f}(-v_a))-2\re(\hat{f}(v_b))=3^{m}.$$
Note that $v_a-v_b,v_b,-v_a$ are pairwise distinct as $v_a\neq -v_b$. In addition,
$$v_a-v_b+v_b-v_a=\bzero.$$

Combining the discussions above, we complete the proof.
\end{proof}

\subsection{A family of minimal ternary linear codes with $w_{\min}/w_{\max} \leq 2/3$}

In this section, we present a family of  minimal ternary linear codes with 
$w_{\min}/w_{\max} \leq 2/3$ with the help of Theorem \ref{thm-1stconminimal}. 
Before doing this, we recall that when $q=3$, the
Lloyd polynomial in (\ref{eqn-lylod-general}) becomes 
\begin{eqnarray}\label{eqn-lylod-ternary}
\Psi_{k}(x,m)=\sum_{t=0}^{k}K_t(x,m)=\sum_{t=0}^k\sum_{j=0}^t (-1)^j2^{t-j} \binom{x}{j} \binom{m-x}{t-j}.
\end{eqnarray}

For a positive integer $k$ with $1\leq k \leq m$, let $S(m,k)$ denote the set of vectors in $\gf(3)^{m}\setminus \{\mathbf{0}\}$ with Hamming weight at most $k$.
It is clear that
$$|S(m,k)|=\sum_{j=1}^{k}2^{j}\binom{m}{j}.$$
Define a function $g_{(m,k)}$ from $\gf(3)^{m}$ to $\gf(3)$ as
\begin{eqnarray}\label{eqn-myfunc}
g_{(m,k)}(x)=
\left\{
\begin{array}{ll}
1 & \mbox{ if } x\in S(m,k),\\
0 & \mbox{ otherwise. }
\end{array}
\right.
\end{eqnarray}

Using $g_{(m,k)}$ to replace the function $f$ in (\ref{eqn-mycode}), we automatically obtain a ternary linear code
$\C_{g_{(m,k)}}$. The parameters and weight distribution of $\C_{g_{(m,k)}}$ are given as follows.

\begin{theorem}\label{thm-weightdistribution}
The ternary code $\C_{g_{(m,k)}}$ has length $3^m-1$, dimension $m+1$, and the weight distribution in Table
\ref{tab-1}, where $\Psi_{k}(x,m)$ is the Lloyd polynomial given by (\ref{eqn-lylod-ternary}).
\end{theorem}
\begin{proof}
From the definition of $g_{(m,k)}$ in Equation (\ref{eqn-myfunc}),
\begin{eqnarray*}
\widehat{g}_{(m,k)}(w)&=&\sum_{x\in \gf(3)^{m}}\zeta_{3}^{g_{(m,k)}(x)-w\cdot x}\\
&=&\sum_{x\in S(m,k)}\zeta_{3}^{1-w\cdot x}+\sum_{x\in \gf(3)^{m}\setminus S(m,k)}\zeta_{3}^{-w\cdot x}\\
&=&\sum_{x\in \gf(3)^{m}}\zeta_{3}^{-w\cdot x}+(\zeta_3-1)\sum_{x\in S(m,k)}\zeta_{3}^{-w\cdot x}
\end{eqnarray*}
for $w\in \gf(3)^{m}$. If $w=\bzero$, then
$$\widehat{g}_{(m,k)}(\bzero)=3^{m}+(\zeta_3-1)\sum_{j=1}^{k}2^{j}\binom{m}{j}$$
and
$$\re\left(\widehat{g}_{(m,k)}(\bzero)\right)=3^{m}-\frac{3}{2}\sum_{j=1}^{k}2^{j}\binom{m}{j}.$$
If $w\neq \bzero$ with $\wt(w)=i$, then by Lemma  \ref{lem-krawweight} we have
\begin{eqnarray*}
\widehat{g}_{(m,k)}(w) = (\zeta_3-1)\sum_{t=1}^{k}K_t(i,m) 
= (\zeta_3-1)(\Psi_{k}(i,m)-1)
\end{eqnarray*}
for $q=3$. Thus
$$\re\left(\widehat{g}_{(m,k)}(w)\right)=-\frac{3}{2}(\Psi_{k}(i,m)-1).$$
Then the weight distribution follows from the proof of Theorem \ref{thm-codeparam}.
\end{proof}

\begin{rem}
It is noticed that the {Lloyd polynomial} $\Psi_{k}(x,m)$ may take the same
value for different values of $x$. Therefore, the set $\C_{g_{(m,k)}}$ in Theorem \ref{thm-weightdistribution}
is a ternary linear code with at most $m+2$ weights. For example, when $k=2$ and $m=5$, 
$$\Psi_{2}(x,5)=\frac{9x^2}{2}-\frac{63x}{2}+51.$$
It is easy to see that $\Psi_{2}(x,5)=6$ for $x\in \{2,5\}$, and 
$\Psi_{2}(x,5)=-3$ for $x\in \{3,4\}$. Thus the set $\C_{g_{(7,2)}}$ in Theorem \ref{thm-weightdistribution}
is a five-weight (instead of seven-weight) linear code. Similarly, when $k=2$ and $m=7$,
it is easily verified that the Lloyd polynomial $\Psi_{2}(x,7)$ takes different values when $x$ runs 
from $1$ to $7$. Therefore, the set $\C_{g_{(7,2)}}$ in Theorem \ref{thm-weightdistribution}
is a nine-weight linear code.  
\end{rem} 

\begin{table}[ht]
\center
\caption{Weight distribution}\label{tab-1}
{
\begin{tabular}{lr}
\hline
Weight $w$    & No. of codewords $A_w$  \\ \hline
$0$          & $1$ \\
$3^{m}-3^{m-1} + \Psi_{k}(i,m)-1$        & $2^{i+1}\binom{m}{i}$ \\
        &  \ $1 \leq i \leq m$ \\
$\sum_{j=1}^{k}2^{j}\binom{m}{j}$  & $2$ \\
$3^{m}-3^{m-1}$      & $3^m-1$ \\
\hline
\end{tabular}
}
\end{table}

\begin{corollary}\label{cor2}
Let $m,k$ be integers with $m\geq 5$ and $2\leq k \leq \lfloor\frac{m-1}{2}\rfloor$.  Then the linear code $\C_{g_{(m,k)}}$ in Theorem \ref{thm-weightdistribution} has parameters
$$\left[3^{m}-1,m+1,\sum_{j=1}^{k}2^{j}\binom{m}{j}\right].$$ Furthermore, $w_{\min}/w_{\max}\leq 2/3$ if and only if
$$3\sum_{j=1}^{k}2^{j}\binom{m}{j}\leq 2(3^{m}-3^{m-1})+2^{k+1}\binom{m-1}{k}-2.$$
\end{corollary}

\begin{proof}
By Table \ref{tab-1}, we denote all the nonzero weights in $\C_{g_{(m,k)}}$ as
\begin{eqnarray*}
\left\{
\begin{array}{ll}
\w(i)=3^{m}-3^{m-1} + \Psi_{k}(i,m)-1,\ 1\leq i \leq m,\\
\w'=\sum_{j=1}^{k}2^{j}\binom{m}{j},\\
\w''=3^{m}-3^{m-1}.
\end{array}
\right.
\end{eqnarray*}
For $1\leq i \leq m$, by Corollary \ref{cor1}, we deduce that
\begin{eqnarray}\label{eqn-proof1}
\nonumber \w(i)&=&3^{m}-3^{m-1} + \Psi_{k}(i,m)-1\\
\nonumber & \geq & 3^{m}-3^{m-1}-2^{k}\binom{m-1}{k}-1\\
\nonumber&=&2\times 3^{m-1}-2^{k}\binom{m-1}{k}-1\\
\nonumber&=&2\times(2+1)^{m-1}-2^{k}\binom{m-1}{k}-1\\
&=&2\times\sum_{j=0}^{m-1}2^{j}\binom{m-1}{j}-2^{k}\binom{m-1}{k}-1. 
\end{eqnarray}
Note that
\begin{eqnarray}\label{eqn-proof2}
 \w' =\sum_{j=1}^{k}2^{j}\binom{m}{j} 
= \sum_{j=1}^{k}2^{j}\binom{m-1}{j}+\sum_{j=1}^{k}2^{j}\binom{m-1}{j-1}.
\end{eqnarray}
Since $2\leq k \leq \lfloor\frac{m-1}{2}\rfloor$, we have 
\begin{eqnarray*}
\lefteqn{ 2\times\sum_{j=0}^{m-1}2^{j}\binom{m-1}{j}-2^{k}\binom{m-1}{k}-1 } \\
&=&\sum_{j=0}^{m-1}2^{j}\binom{m-1}{j}+\sum_{j=0}^{m-1}2^{j}\binom{m-1}{j}-2^{k}\binom{m-1}{k}-1\\
&>& \sum_{j=0}^{k}2^{j}\binom{m-1}{j}+\sum_{j=1}^{m-1}2^{j}\binom{m-1}{j}-2^{k}\binom{m-1}{k}-1\\
&=&\sum_{j=1}^{k}2^{j}\binom{m-1}{j}+\sum_{j=1}^{k-1}2^{j}\binom{m-1}{j}+\sum_{j=k+1}^{m-1}2^{j}\binom{m-1}{j}\\
&>&\sum_{j=1}^{k}2^{j}\binom{m-1}{j}+\sum_{j=1}^{k-1}2^{j}\binom{m-1}{j-1}+2^{k+1}\binom{m-1}{k+1}\\
&>&\sum_{j=1}^{k}2^{j}\binom{m-1}{j}+\sum_{j=1}^{k-1}2^{j}\binom{m-1}{j-1}+2^{k}\binom{m-1}{k-1}\\
&=&\sum_{j=1}^{k}2^{j}\binom{m-1}{j}+\sum_{j=1}^{k}2^{j}\binom{m-1}{j-1}.
\end{eqnarray*}
It then follows from  Equations (\ref{eqn-proof1}) and (\ref{eqn-proof2}) that
\begin{eqnarray*}
\w(i)>\w'\mbox{ for all }1\leq i \leq m.
\end{eqnarray*}
From the discussions above, we also have
\begin{eqnarray*}
\w''=3^{m}-3^{m-1}>\w'.
\end{eqnarray*}
Hence, the minimum Hamming weight of $\C_{g_{(m,k)}}$ is given by 
$ 
w_{\min}=\w'.
$ 
According to Corollary \ref{cor1}, the maximum Hamming weight of $\C_{g_{(m,k)}}$ is given by
$$
w_{\max}=\w(1)=3^{m}-3^{m-1} + \Psi_{k}(1,m)-1=3^{m}-3^{m-1}+2^{k}\binom{m-1}{k}-1.
$$
This completes the proof.
\end{proof}

The following lemma will be used to prove the minimality of the linear code in Theorem \ref{thm-weightdistribution}.

\begin{lemma}\label{lem-last}
Let $m,k$ be integers with $m\geq 5$ and $2\leq k \leq \lfloor\frac{m-1}{2}\rfloor$.  Then we have
$$\sum_{j=1}^{k}2^{j}\binom{m}{j}\neq -2(\Psi_{k}(i,m)-1)\mbox{ for all }1\leq i \leq m.$$
\end{lemma}
\begin{proof}
Note that the inequality $$\sum_{j=1}^{k}2^{j}\binom{m}{j}\neq -2(\Psi_{k}(i,m)-1)\mbox{ for all }1\leq i \leq m$$ is equivalent to
$$-\sum_{j=1}^{k}2^{j-1}\binom{m}{j}-\Psi_{k}(i,m)\neq -1\mbox{ for all }1\leq i \leq m.$$
By Lemma \ref{lem-Lloydpolynomial}, it is sufficient to prove that
\begin{eqnarray}\label{ineqn-proof-lastlemma}
-\sum_{j=1}^{k}2^{j-1}\binom{m}{j}-K_{k}(i-1,m-1)\neq -1 \mbox{ for all }1\leq i \leq m.
\end{eqnarray}
Note that
\begin{eqnarray*}
\lefteqn{ -\sum_{j=1}^{k}2^{j-1}\binom{m}{j}-K_{k}(i-1,m-1) } \\
&=&-\sum_{j=1}^{k}2^{j-1}\binom{m-1}{j}-\sum_{j=1}^{k}2^{j-1}\binom{m-1}{j-1}-K_{k}(i-1,m-1)\\
&<&-2^{k-1}\binom{m-1}{k}-K_{k}(i-1,m-1).
\end{eqnarray*}
This means that (\ref{ineqn-proof-lastlemma}) holds provided that
$$-2^{k-1}\binom{m-1}{k}-K_{k}(i-1,m-1)<-1\mbox{ for all }1\leq i \leq m.$$
In fact, by the Vandermonde convolution formula and the definition of the Krawthouk polynomial, we have
\begin{eqnarray*}
\lefteqn{ -2^{k-1}\binom{m-1}{k}-K_{k}(i-1,m-1) } \\
&=&-2^{k-1}\sum_{j=0}^{k}\binom{i-1}{j}\binom{m-i}{k-j}-\sum_{j=0}^k (-1)^j2^{k-j} \binom{i-1}{j} \binom{m-i}{k-j}\\
&=&-\sum_{j=0}^k  \binom{i-1}{j} \binom{m-i}{k-j}\left(2^{k-1}+(-1)^j2^{k-j}\right)\\
&<&-1,
\end{eqnarray*}
where we used the fact that $2\leq k \leq \lfloor\frac{m-1}{2}\rfloor$ and $m\geq 5$. This completes the proof of 
this lemma. 
\end{proof}

The following result shows that the linear code in Theorem \ref{thm-weightdistribution} 
is minimal and violates the Ashikhmin-Barg condition in many cases.  

\begin{theorem}\label{mainth}
Let $m,k$ be integers with $m\geq 5$ and $2\leq k \leq \lfloor\frac{m-1}{2}\rfloor$. Then the linear code $\C_{g_{(m,k)}}$ is minimal and has parameters
$$\left[3^{m}-1,m+1,\sum_{j=1}^{k}2^{j}\binom{m}{j}\right].$$ Furthermore, $w_{\min}/w_{\max}\leq 2/3$ if and only if
$$3\sum_{j=1}^{k}2^{j}\binom{m}{j}\leq 2(3^{m}-3^{m-1})+2^{k+1}\binom{m-1}{k}-2.$$
\end{theorem}

\begin{proof}
According to  Corollary \ref{cor2}, we only need to prove that $\C_{g_{(m,k)}}$ is minimal.
From the proof of Theorem \ref{thm-weightdistribution}, we have
\begin{eqnarray}\label{sys}
\re\left(\widehat{g}_{(m,k)}(w)\right)=
\left\{
\begin{array}{ll}
3^{m}-\frac{3}{2}\sum_{j=1}^{k}2^{j}\binom{m}{j} & \mbox{ if } w=\bzero,\\
-\frac{3}{2}(\Psi_{k}(i,m)-1) & \mbox{ if }\wt(w)=i>0.
\end{array}
\right.
\end{eqnarray}
Theorem \ref{thm-1stconminimal} implies that $\C_{g_{(m,k)}}$ is minimal if and only if
\begin{eqnarray}\label{ineqn-1}
\re(\hat{g}_{(m,k)}(w_1))+\re(\hat{g}_{(m,k)}(w_2))-2\re(\hat{g}_{(m,k)}(w_3))\neq3^{m}
\end{eqnarray}
and
\begin{eqnarray}\label{ineqn-2}
\re(\hat{g}_{(m,k)}(w_1))+\re(\hat{g}_{(m,k)}(w_2))+\re(\hat{g}_{(m,k)}(w_3))\neq3^{m}
\end{eqnarray}
for any pairwise distinct  vectors $w_1,w_2,w_3\in \gf(3)^{m}$ satisfying $w_1+w_2+w_3=\bzero$.
We distinguish between the following two cases to show that (\ref{ineqn-1}) and (\ref{ineqn-2}) hold for the claimed vectors.

\subsubsection*{Case 1: Assume that one of $w_1,w_2,w_3$ is $\bzero$.}

We firstly consider Inequality (\ref{ineqn-2}). Without loss of generality, we assume that $w_1=\bzero$ and then $w_2=-w_3\neq \bzero$, where $\wt(w_2)=\wt(w_3)=i$ and $1\leq i \leq m$. Then by Equation (\ref{sys}), Inequality  (\ref{ineqn-2}) is equivalent to
    $$\sum_{j=1}^{k}2^{j}\binom{m}{j}\neq -2(\Psi_{k}(i,m)-1)\mbox{ for all }1\leq i \leq m,$$
     which  holds by Lemma \ref{lem-last}. Next we consider Inequality (\ref{ineqn-1}) in two cases.
    \begin{enumerate}
    \item If $w_3=\bzero$, then $\wt(w_1)=\wt(w_2)=i$ with $1\leq i \leq m$. Then by  Equation (\ref{sys}), Inequality  (\ref{ineqn-1}) is equivalent to
        \begin{eqnarray}\label{ineqn-3}\sum_{j=0}^{k}2^{j}\binom{m}{j}-\Psi_{k}(i,m)\neq 3^{m}\mbox{ for all }1\leq i \leq m.\end{eqnarray}
        Due to Corollary \ref{cor1} for $q=3$, we deduce that
        \begin{eqnarray*}
        \left|\sum_{j=0}^{k}2^{j}\binom{m}{j}-\Psi_{k}(i,m)\right|&\leq & \sum_{j=0}^{k}2^{j}\binom{m}{j}+|\Psi_{k}(i,m)|\\
        &\leq &\sum_{j=0}^{k}2^{j}\binom{m}{j}+2^{k}\binom{m-1}{k}\\
        &<&\sum_{j=0}^{k}2^{j}\binom{m}{j}+2^{k}\binom{m-1}{k}+2^{k}\binom{m-1}{k+1}\\
        &=&\sum_{j=0}^{k}2^{j}\binom{m}{j}+2^{k}\binom{m}{k+1}\\
        &<&\sum_{j=0}^{k}2^{j}\binom{m}{j}+2^{k+1}\binom{m}{k+1}\\
        &<&\sum_{j=0}^{m}2^{j}\binom{m}{j}=3^{m}.
        \end{eqnarray*}
        Thus Inequality (\ref{ineqn-3}) holds and then Inequality  (\ref{ineqn-1}) holds.
    \item If one of $\wt(w_1)$ and $\wt(w_2)$ is $\bzero$, we assume that $w_1=\bzero$ without loss of generality. Then $\wt(w_2)=\wt(w_3)=i$ with $1\leq i \leq m$. Then by Equation (\ref{sys}), Inequality  (\ref{ineqn-1}) is equivalent to
                \begin{eqnarray}\label{ineqn-4}\sum_{j=0}^{k}2^{j}\binom{m}{j}\neq\Psi_{k}(i,m)\mbox{ for all }1\leq i \leq m.\end{eqnarray}
                Due to Corollary \ref{cor1} for $q=3$, we have
                $$\Psi_{k}(i,m)\leq 2^{k}\binom{m-1}{k}<2^{k}\binom{m}{k}<\sum_{j=0}^{k}2^{j}\binom{m}{j}.$$
                Thus Inequality (\ref{ineqn-4}) holds and then Inequality  (\ref{ineqn-1}) holds.
    \end{enumerate}
    In this case, (\ref{ineqn-1}) and (\ref{ineqn-2}) follow from the discussion above.

\subsubsection*{Case 2: Assume that all $w_1,w_2,w_3$ are nonzero.}

Due to Corollary \ref{cor1} for $q=3$ and Equation (\ref{sys}), we derive that
\begin{eqnarray*}
\left|\re(\hat{g}_{(m,k)}(w_1))+\re(\hat{g}_{(m,k)}(w_2))-2\re(\hat{g}_{(m,k)}(w_3))\right|\leq 6\times 2^{k}\binom{m-1}{k}+6
\end{eqnarray*}
and
\begin{eqnarray*}
\left|\re(\hat{g}_{(m,k)}(w_1))+\re(\hat{g}_{(m,k)}(w_2))+\re(\hat{g}_{(m,k)}(w_3))\right|\leq \frac{9}{2}\times 2^{k}\binom{m-1}{k}+\frac{9}{2}.
\end{eqnarray*}
To show that Inequalities (\ref{ineqn-1}) and (\ref{ineqn-2}) hold, it is sufficient to show that
$$6\times 2^{k}\binom{m-1}{k}+6<3^{m},$$
which is equivalent to
\begin{eqnarray*}
2^{k+1}\binom{m-1}{k}+2<3^{m-1}.
\end{eqnarray*}
Since $2\leq k \leq \lfloor\frac{m-1}{2}\rfloor$, we have
\begin{eqnarray*}
2^{k+1}\binom{m-1}{k}+2&=&2^{k}\binom{m-1}{k}+2^{k}\binom{m-1}{k}+2\\
&<&2^{k}\binom{m-1}{k}+2^{k+1}\binom{m-1}{k+1}+2\\
&<&\sum_{j=0}^{m-1}2^{j}\binom{m-1}{j}=3^{m-1}.
\end{eqnarray*}
Thus Inequalities (\ref{ineqn-1}) and (\ref{ineqn-2}) hold in this case.

Summarizing the discussion above completes the proof of this theorem.
\end{proof}

As a corollary of Theorem \ref{mainth}, one can easily derive the following.

\begin{corollary}\label{cor-infinite}
Let $k=2,\ m \geq 5$. Then  $\C_{g_{(m,2)}}$ in Theorem \ref{mainth} is a  minimal code with parameters
$$\left[3^{m}-1,m+1,\sum_{j=1}^{2}2^{j}\binom{m}{j}\right].$$
Furthermore,
$w_{\min}/w_{\max} < 2/3$.
\end{corollary}

Corollary \ref{cor-infinite} demonstrates that the codes in Theorem \ref{mainth} contain an infinite family of ternary minimal linear codes with $w_{\min}/w_{\max} \leq 2/3$.  
We conclude this section with two examples computed by Magma, which are consistent with our theoretical results.

\begin{example}\label{exa-1}
The set $\C_{g_{(5,2)}}$ in Theorem \ref{mainth} is a  minimal code with parameters
$\left[242,\ 6,\ 50 \right]$ and weight enumerator
$$
1+  2z^{50} + 320z^{158} + 242z^{162} + 144z^{167} + 20 z^{185},
$$
where
$w_{\min}/w_{\max} = 50/185 < 2/3$. 
\end{example}

\begin{example}
The set $\C_{g_{(7,2)}}$ in Theorem \ref{mainth} is a  minimal code with parameters
$\left[2186,\ 8,\ 98 \right]$ and weight enumerator
$$
1+  2z^{98} + 1344^{1451} + 1120^{1454} + 896^{1457} + 2186 z^{1458}+560 z^{1466}+256 z^{1472}+168 z^{1487}+28 z^{1517},
$$
where
$w_{\min}/w_{\max} = 98/1517 < 2/3$.
\end{example}

\section{Concluding remarks}\label{sec-concluison}

In this paper, we derived a necessary and sufficient condition for linear codes over  
finite fields to be minimal. This condition generalizes the construction for the binary 
case given in  \cite{DHZ}. It enabled us to obtain a family of ternary minimal linear 
codes violating the   Ashikhmin-Barg condition. This is the first infinite family of nonbinary minimal linear codes violating the Ashikhmin-Barg condition. It would be 
interesting to construct more such infinite families of nonbinary minimal 
linear codes. 
The reader is cordially invited to join this adventure.

Finally, we mention that the construction of the ternary code in (\ref{eqn-mycode}) 
can be generalised to obtain $q$-ary codes, and the function $g_{m,k}(x)$ of 
(\ref{eqn-myfunc}) is also a function from $\gf(q)$ to $\gf(q)$ when the set 
$S(m,k)$ is defined to be the set of all vectors in $\gf(q)^m \setminus \{\bzero\}$ 
with Hamming weight at most $k$. 
In this case, the code $\C_{g(m,k)}$ is a linear code over $\gf(q)$. Experimental 
data show that the code $\C_{g(m,k)}$ over $\gf(q)$ could be minimal under certain conditions. However, it seems very hard to work out similar results for the general 
case $q$. We were able to handle only the ternary case in this paper.

\end{document}